\documentclass[fleqn,12pt]{article}

\usepackage{graphicx}
\usepackage[figuresright]{rotating}
\usepackage{setspace}

\newtheorem{theorem}{Theorem}

\newtheorem{conjecture}{Conjecture}

\usepackage{espcrc1}

\newenvironment{proof}[1][Proof.]{\begin{trivlist}
\item[\hskip \labelsep {\bfseries #1}]}{\end{trivlist}}

\usepackage{amsmath}
\usepackage{amsfonts}

\begin{document}

\title{Bricks and conjectures of Berge, Fulkerson and Seymour}

\author{Vahan V. Mkrtchyan\address[MCSD]{Paderborn Institute for Advanced Studies in Computer Science and Engineering,
Paderborn University, Warburger Str. 100, 33098 Paderborn, Germany}%
\thanks{The author is supported by a fellowship from Heinrich Hertz-Stiftung}
\thanks{email: vahanmkrtchyan2002@\{ysu.am, ipia.sci.am,
yahoo.com\}}
                        and
        Eckhard Steffen \addressmark[MCSD]\thanks{email: es@upb.de}}


\runtitle{Bricks and conjectures of Berge, Fulkerson and Seymour}
\runauthor{Vahan Mkrtchyan, Eckhard Steffen}
\maketitle

\begin{abstract}
An $r$-graph is an $r$-regular graph where every odd set of vertices is connected by at least $r$ edges to the rest of the graph. Seymour conjectured that any $r$-graph is $r+1$-edge-colorable, 
and also that any $r$-graph contains $2r$ perfect matchings such that each edge belongs to two of them. We show that the minimum counter-example to either of these conjectures is a brick.
Furthermore we disprove a variant of a conjecture of Fan, Raspaud. 
\end{abstract}

\section{Introduction and definitions}

We consider finite graphs $G=(V,E)$ with vertex set $V$ and edge set $E$. The graphs might have multiple 
edges but no loops, and throughout this paper we assume that the graphs under consideration are connected (otherwise we will mention this explicitely). 
Terms and concepts that we do not define can be found in Chapter 3 of \cite{Handbook}.

A perfect matching of a graph $G$ is a matching covering all vertices of $G$, and $\theta(G)$ denotes the number of odd components of $G$.  
Tutte characterized the graphs with perfect matching.

\begin{theorem}\label{Tutte}(Tutte) A graph $G$ has a perfect matching if and only if $\theta(G-X) \leq |X|$, for each $X\subseteq V(G)$.
\end{theorem}

A class of graphs possessing a perfect matching is the class of $r$-graphs \cite{Seymour_1979}, which are $r$-regular graphs with  $|\partial(X)|\geq r$ for every odd $X\subseteq V(G)$, 
where $\partial(X)$ is the set of edges of $G$ with precisely one end in $X$. The following conjectures are due to Seymour \cite{Seymour_1979}.

\begin{conjecture}\label{SeymourColoring} Any $r$-graph is $(r+1)$-edge-colorable.
\end{conjecture}

\begin{conjecture}\label{SeymourPerfects} Any $r$-graph contains $2r$ perfect matchings such that each edge belongs to precisely two of them.
\end{conjecture}

Conjecture \ref{SeymourPerfects} was first formulated by Berge, Fulkerson for $r=3$ (Berge Fulkerson Conjecture).  
The Berge Fulkerson Conjecture implies the following conjecture made by Fan, Raspaud \cite{FanRaspaud}.

\begin{conjecture}\label{FanRaspaud} Any $3$-graph contains three perfect matchings $F_1,F_2,F_3$ such that $F_1\cap F_2 \cap F_3=\emptyset$.
\end{conjecture}

Conjecture \ref{SeymourPerfects} implies that every $r$-graph has $2r$ perfect matchings such that any three of them have an empty intersection. In the context of the conjectures
of Seymour, Berge, Fulkerson and of Fan, Raspaud, the following statement seems to be a natural generalization of Conjecture \ref{FanRaspaud}.

\begin{conjecture}\label{FanRaspaud_r} For $r \geq 3$, any $r$-graph contains $r$ perfect matchings such that the intersection of any three of them is empty.
\end{conjecture}

As the Berge, Fulkerson Conjecture implies Conjecture \ref{FanRaspaud}, Conjecture \ref{SeymourPerfects} implies Conjecture \ref{FanRaspaud_r}.
This paper studies the structure of minimum counter example to Conjectures \ref{SeymourColoring} and \ref{SeymourPerfects}, 
and a variation of \ref{FanRaspaud_r}.

A graph $G=(V,E)$ is matching covered if every edge belongs to a perfect matching. It is factor-critical if $G-u$ has a perfect matching for every vertex $u$, and it 
is bi-critical if $G-u-v$ has a perfect matching for every pair of vertices $u$ and $v$.  A barrier in a matching covered graph is a set $X \subseteq V$
such that $\theta(G-X) = |X|$. Note that a single vertex in a matching covered graph is a barrier. A non-bipartite, bi-critical and $3$-vertex-connected graph is a brick. 
This paper proves that a minimal counter example to Conjectures \ref{SeymourColoring} and \ref{SeymourPerfects} is a brick. 

There are almost no results on these conjectures, and one might think about some variations of the original conjectures. 
It is a natural question whether a perfect matching can be fixed in Conjectures \ref{SeymourPerfects} or \ref{FanRaspaud}, say: 
Let $G$ be an $r$-graph ($ \geq 3$) and $F$ a perfect matching of $G$, then $G$ has $r-1$ perfect matchings $F_1,F_2, \dots ,F_{r-1}$ such that 
for any $1 \leq i < j \leq r-1$ the intersection $F \cap F_i \cap F_j$ is empty. We will show that this this is not true for any odd $r$.

\section{The main results}

Let $G=(V,E)$ be an $r$-graph. An odd set $X \subseteq V$ with $|\partial(X)|=r$ is non-trivial, if $|X| \not = 1, |V|-1$.
We start with the following theorem. 

\begin{theorem}\label{decomposition}Any $r$-graph $G=(V,E)$ satisfies at least one of the following conditions:
\begin{enumerate}
	\item [(1)] $G$ is bipartite;
	\item [(2)] there is non-trivial odd $X\subseteq V(G)$, such that $|\partial(X)|=r$;
	\item [(3)] $G$ is bi-critical.
\end{enumerate}
\end{theorem}

\begin{proof} Clearly, an $r$-graph is matching covered (see Exercise 3.4.4 in \cite{Lov} p. 113), and let
$S$ be any maximal barrier with $|S|=k$. Then, $G-S$ has $k$ componets $G_1,...,G_k$, which are factor-critical (see Lemma 4.6. in \cite{Handbook} p. 198). 
Since $G_1,...,G_k$ are odd, it follows that $|\partial(G_i)|\geq r$, for $i=1,...,k$. On the other hand, the set $S$ can receive at most $r|S|=rk$ edges, 
thus $S$ is an independent set and $|\partial(G_i)|=r$, for $i=1,...,k$. 

Now, consider the following two cases:

Case 1: There is a maximal barrier $S$ with $|S|>1$. 
If the components of $G-S$ are isolated vertices, then $G$ is bipartite and hence (1). If there is a component $G_i$ of $G-S$, 
that is not an isolated vertex, then $|\partial(G_i)|=r$ and $G$ satisfies the condition (2) of the theorem.

Case 2: for any maximal barrier $S$, we have $|S|=1$.
Let $S=\{u\}$. $G-S=G-u$ is factor-critical, which means that for any vertex $v$, the graph $G-u-v$ has a perfect matching, thus $G$ satisfies the condition (3).
$\square$
\end{proof}

\begin{theorem}\label{BrickTheorem} Let $G=(V,E)$ be a bi-critical non-bipartite $r$-graph, that contains no non-trivial odd set $X \subseteq V$ with $|\partial(X)|=r$. Then $G$ is a brick.
\end{theorem}

\begin{proof} Since $G$ is bi-critical and non-bipartite, and $r$-graphs are $2$-vertex-connected, it suffices to show that $G$ does not have a $2$-vertex-cut. 

Suppose that there is one and let $u,v$ be the vertices of the $2$-cut. Since $G$ is bi-critical, the graph $G-u-v$ has a perfect matching, thus all connected 
components $G_1,...,G_k$ of $G-u-v$ are even. Let $x$ denote the number of edges of $G$ that connect the vertices $u$ and $v$. For $i=1,...,k$ let $y_i$ and $z_i$ 
be the number of edges that connects $u$ and $v$ to $G_i$, respectively. Clearly,
\begin{equation*}
x+y_1+...+y_k=r=x+z_1+...+z_k
\end{equation*}

The sets $V(G_i)\cup \{u\}$ and $V(G_i)\cup \{v\}$ are odd. Thus $|\partial(V(G_i)\cup \{u\})|> r$ and $|\partial(V(G_i)\cup \{v\})|> r$. It follows that
\begin{equation*}
z_i+x+y_1+...+y_{i-1}+y_{i+1}+...+y_k>r=x+y_1+...+y_k
\end{equation*}
and
\begin{equation*}
y_i+x+z_1+...+z_{i-1}+z_{i+1}+...+z_k>r=x+z_1+...+z_k,
\end{equation*}
and hence
$z_i>y_i$ and $y_i>z_i$, which is a contradiction.
$\square$
\end{proof}

\begin{theorem} A minimum counter-example to either of conjectures \ref{SeymourColoring} and \ref{SeymourPerfects} is a brick.
\end{theorem}

\begin{proof} Let $G$ be a minimum counter-example to conjecture \ref{SeymourColoring}. Since bipartite $r$-regular graphs are $r$-edge-colorable, $G$ is not bipartite. Next, we show that there is no non-trivial odd $X\subseteq V(G)$, such that $|\partial(X)|=r$. 

Suppose that there is one. Then consider the two graphs $G_1$ and $G_2$ that are obtained from $G$ by contracting $X$ and $V(G)\backslash X$ to a vertex, respectively. Clearly $G_1$ and $G_2$ are $r$-graphs; moreover since they are smaller than $G$, they are $(r+1)$-edge-colorable. Now, it is not hard to see that an $(r+1)$-edge-coloring of $G$ can be obtained from those of $G_1$ and $G_2$, which would contradict the choice of $G$.

Thus, $G$ contains no non-trivial odd $X\subseteq V(G)$, with $|\partial(X)|=r$. Theorem \ref{decomposition} implies that $G$ is bi-critical, and hence $G$ is a brick by 
Theorem \ref{BrickTheorem}.

The proof for conjecture \ref{SeymourPerfects} follows the same lines.
$\square$
\end{proof}

An $r$-graph $G$ is unslicable if for any perfect matching $F$, the graph $G-F$ is not an $(r-1)$-graph. Rizzi \cite{Rizzi} constructed unslicable $r$-graphs for every $r \geq 3$. 

\begin{theorem} For every $k \geq 1$, there is a $(2k+1)$-graph $G$ with perfect matching $F$, such that for any $2k$ perfect matchings $F_1, \dots, F_{2k}$ there are 
$1 \leq i < j \leq 2k$ such that $F \cap F_i \cap F_j \not =\emptyset$.
\end{theorem}

\begin{proof} Let $r = 2k+1$ $(k \geq 1)$ be an odd number and $H$ an unslicable $r$-graph. Let $C$ be a (multi-) cycle of length $r$, where every edge has multiplicity $k$, 
i.e. $C$ is $2k$-regular. Replace every vertex $v$ of $H$ by a copy $C_v$ of $C$ to obtain a $(2k+1)$-regular graph $G =(V,E)$. Note that the {\em old} edges of $H$ in $G$ form a 
perfect matching $F$ of $G$. 

We first show that $G$ is an $(2k+1)$-graph. Clearly $|V|$ is even. Let $X \subseteq V$ be an odd set and assume that $|\partial_G(X)| < 2k+1$. If $\partial_G(X)$ contains an edge of 
a cycle $C_v$, then it contains at least $2k$ of them, i.e. $|\partial_G(X)| = 2k$. This implies that $v$ is a cut vertex in $H$, contradicting the fact that $H$ is 2-vertex-connected. Thus 
$\partial_G(X)$ contains at least one old edge of $H$ and hence $|\partial_G(X)| \geq 2k+1$, contradicting our assumption. 

Thus we may assume that $\partial_G(X)$ contains only old edges of $H$. Let $X^-$ be the subset of $G$ from which $X$ is obtained in the transfornmation from $H$ to $G$. 
$X$ is an odd set in $G$ and hence 
$X^-$ is an odd set in $H$. But then $|\partial_H(X)| < 2k+1 \leq r$, contradicting the fact that $H$ is an $r$-graph. Thus $G$ is a $(2k+1)$-graph. 

Now assume that $G$ has $r-1$ perfect matchings $F_1, \dots ,F_{r-1}$ such that $F \cap F_i \cap F_j = \emptyset$, for any $1 \leq i < j \leq r-1$. 
Consider $\partial_G(C_v)$. Since it is an odd cut, it follows
that  $\partial_G(C_v) \cap F_i \not = \emptyset$, for $i=1,2,...,r-1$. Furthermore $\partial_G(C_v) \subseteq F$ by the choice of $F$. If there is a perfect matching, say $F_1$ such that 
$|\partial_G(C_v) \cap F_1| > 1$, then, since we assume that $F \cap F_1 \cap F_i = \emptyset$ ($i > 1$), the remaining $r-2$ perfect matching $F_2 \dots, F_{r-1}$ share $r-3$ edges.
Thus there are $i\not=j$, such that the intersection of $F$, $F_i$, and $F_j$ is not empty. 
Thus every perfect matching of $F_1,\dots,F_{r-1}$ contains precisely one edge of $\partial_G(C_v)$, and they are pairwise disjoint on $\partial_G(C_v)$. 
But this implies that they induce $r-1$ pairwise disjoint perfect matchings $F_1^-, \dots, F_{r-1}^-$ on $H$. Thus $H$ contains an $r-1$-graph, contradicting the fact that
$H$ is unslicable. $\square$
\end{proof}


\begin{thebibliography}{99}

\bibitem{FanRaspaud} G. Fan and A. Raspaud, Fulkerson's Conjecture and circuit covers, J. Combin. Theory Ser. B 61 (1994) 133-138.

\bibitem{Handbook} R.L. Graham, M. Gr\"otschel, L. Lov\'asz, Handbook of Combinatorics, Vol. 1, 2, Elsevier, Amsterdam, 1995

\bibitem{Lov} L. Lov\'asz, M.D. Plummer, Matching theory, Ann. Discrete Math. 29 (1986)

\bibitem{Rizzi} R. Rizzi, Indecomposable $r$-graphs and other counterexamples, J. Graph Theory 32 (1999) 1-15

\bibitem{Seymour_1979} P. D. Seymour, On multi-colourings of cubic graphs, and conjectures of Fulkerson and Tutte, Proc. London Math. Soc. 38 (1979) 423-460



\end{thebibliography}
\end{document}